\documentclass[letterpaper, 10 pt, conference]{ieeeconf}  

\IEEEoverridecommandlockouts                              
\usepackage{flushend}
\usepackage{graphicx}
\usepackage{tikz}
\usepackage{color}
\usepackage{amsmath, amsfonts, amssymb, mathrsfs}
\usepackage[amsmath,thmmarks]{ntheorem}
\usepackage{subfigure}
\usepackage{epstopdf}
\usepackage{pgfplots}
\usepackage{arydshln}
\usepackage{mathabx}
\usepackage{nth}
\usepackage{cite}
\usepackage{algorithm}
\usepackage{algorithmic}
\usepackage[T1]{fontenc}

\makeatletter
\renewcommand*\env@matrix[1][*\c@MaxMatrixCols c]{%
  \hskip -\arraycolsep
  \let\@ifnextchar\new@ifnextchar
  \array{#1}}

\theoremstyle{plain}
\theoremheaderfont{\itshape\bfseries}
\theorembodyfont{\normalfont\itshape}
\theoremseparator{:}
\theoremsymbol{}
\newtheorem{theorem}{Theorem}

\newtheorem{corollary}[theorem]{Corollary}
\newtheorem{proposition}[theorem]{Proposition}

\newtheorem{assum}{Assumptions}
\newtheorem{definition}{Definition}

\theoremstyle{nonumberplain}

\theoremstyle{plain}
\theoremheaderfont{\itshape\bfseries}
\theorembodyfont{\normalfont}
\theoremseparator{:}
\theoremsymbol{}
\newtheorem{remark}{Remark}

\theoremstyle{plain}
\theoremheaderfont{\itshape\bfseries}
\theorembodyfont{\normalfont}
\theoremseparator{:}

\theoremstyle{nonumberplain}
\theoremsymbol{\rule{1ex}{1ex}}

\newlength\fheight
\newlength\fwidth
\newcommand{\abs}[1]{\left| #1 \right|}
\newcommand{\norm}[1]{\| #1 \|}
\newcommand{\inn}[2]{\langle #1,#2 \rangle}
\newcommand{\lrbrace}[1]{\left\{ #1 \right\}}

\newcommand{\normi}[1]{{\left\vert\kern-0.25ex\left\vert\kern-0.25ex\left\vert #1 
    \right\vert\kern-0.25ex\right\vert\kern-0.25ex\right\vert}}
\newcommand{\inni}[2]{{\langle\kern-0.25ex\langle #1,#2
   \rangle\kern-0.25ex\rangle}}

\def\X{\mathcal{X}}

\def\nat{ \mathbb{N} }								

\def\real{ \mathbb{R} }								

\def\V{\mathcal{V}}

\def\X{\mathcal{X}}

\DeclareMathOperator{\VI}{\text{VI}}
\DeclareMathOperator{\VS}{\text{VS}}
\DeclareMathOperator{\SOL}{\text{SOL}}

\DeclareMathOperator*{\argmax}{arg\,max}

\newcommand\restr[2]{{
  \left.\kern-\nulldelimiterspace 
  #1 
  \right|_{#2} 
  }}

\def\Z{\mathcal{Z}}

\title{\LARGE \bf
	Semi-Decentralized Coordinated Online Learning for Continuous Games with Coupled Constraints via Augmented Lagrangian
}

\author{Ezra Tampubolon and Holger Boche
	\thanks{E. Tampubolon and H. Boche is with the Department of Electrical and Computer Engineering,
		Technical University of Munich,  80333 Munich, Germany
		{\tt\small \{ezra.tampubolon,boche\}@tum.de}}%
}

\begin{document}
%
\maketitle
%
\thispagestyle{empty}
\pagestyle{empty}
\begin{abstract}
We consider a class of concave continuous games in which the corresponding admissible strategy profile of each player underlies affine coupling constraints. We propose a novel algorithm that leads the corresponding population dynamic toward Nash equilibrium. This algorithm is based on a mirror ascent algorithm, which suits with the framework of no-regret online learning, and on the augmented Lagrangian method. The decentralization aspect of the algorithm corresponds to the aspects that the iterate of each player requires the local information of about how she contributes to the coupling constraints and the price vector broadcasted by a central coordinator. So each player need not know about the population action. Moreover, no specific control by the central coordinator is required. We give a condition on the step sizes and the degree of the augmentation of the Lagrangian, such that the proposed algorithm converges to a generalized Nash equilibrium.
\end{abstract}
%

\section{Introduction}
Competitive non-cooperative selfish agents appear as a model in a vast number of applications (see also \cite{Scutari2012}) such as smart grid \cite{Mohsenian2010,Saad2012,Deng2014,Li2016,Li2016(2),Grammatico2015,Ma2013,Parise2014,Ma2016,Grammatico20162}, competitive markets \cite{Li2015}, and congestion control for networks \cite{Barrera2015}. The famous concept of non-cooperative continuous game theory is
suited to analyze such applications: The typical setting is that a set of agents repeatedly interact with each other, in the sense that at time $t$ the payoff/reward of an agent depends not only on his action but also on the joint action of all other agents which is not visible for him.

Given the uncertainty faced by one agent about the joint action of the others, he has therefore to choose his action in an online manner aiming to optimize his time-variant reward (for detailed discussions see \cite{Belmaga2018} and references therein). A reasonable assumption on their behavior is that they apply the no-regret policy (see, e.g., \cite{Shalev-Shwartz2012,Belmaga2018}) known in the literature of online learning. The canonical class of no-regret policies in the black-box environment, i.e., in the environment where no further assumptions on the utility functions $(u_{i})_{i}$ other than concavity, 
is the so-called online mirror ascent which is a canonical extension of famous mirror ascent algorithm. Mirror ascent consists at each time instance the gradient step in the dual space and the "mirror" step which maps it back to the feasible primal region.

In applications, there is often a coupling between the ideal agent's action set and the actual other agents' action. Some examples are TCP (congestion) control problem \cite{Low1}, where the transmit rates of the agents underly capacity link constraints, the problem of charging of electrical vehicles \cite{Ma2013}, where a threshold of the total demand of the units of power of the agents is available, and MIMO interference systems (see e.g. \cite{Scutari2012} and references therein), where the transmission strategies of the secondary users (agents) in the form of power allocation vectors over the sub-carriers underly sum constraint. Notice that in all the previously mentioned applications, the coupled constrained is of affine form. 

In this work, we consider a novel mirror ascent based algorithms for concave games that can handle coupled affine constraints. In each time step, each agents executes a mirror ascent update which requires locally available first-order information of their utility function and the price vector broadcasted by the central coordinator. The latter is updated by the central coordinator using the augmented-Lagrangian-based update. We give a sufficient condition on the (non-adaptive) step size sequences of the agents and of the central coordinator and on the augmentation of the Lagrangian such that the proposed algorithm converges to a (variationally stable) Nash equilibrium.   

As a matter of course, the proposed algorithm can also be used by a system designer to design agent control algorithms in order to generate a desirable collective behavior in the case that the latter coincides with the Nash equilibrium of the considered coupled constrained game. This is done by correspondingly designing localized agent utility function and designing a related coordination strategy. According to our recommendation, the system designer might realize the latter based on iterative communications with a central coordinator that can gather and broadcast information from/to the population. The motivation for this semi-decentralized approach arises from the privacy demand between the agents, and the computational intractability of a fully centralized solution.

\subsubsection*{Relation to Prior Works} 
\cite{Mertikopoulos2018} provides, among other things, a deterministic analysis of the online mirror ascent algorithm for games with a continuous action set. In contrast to this work, we consider games with continuous action sets which underlies in addition to coupling constraint so that the admissible set of population strategy profile is not necessarily of product structure. We have to modify the decentralized algorithm given in \cite{Mertikopoulos2018} and make use of a central coordinator in order to handle such additional constraints. For this reason, our work is an extension of the deterministic result given in \cite{Mertikopoulos2018}.

\cite{Paccagnan2017} gives the most recent semi-decentralized first-order algorithm for finding equilibrium and handling coupled constraints.There, the authors leverage mostly from the fixed-point method for finding the solution of variational inequality (see e.g., Chapter 12 in \cite{Facchinei1}), which results in (Euclidean-)projection-based algorithm. Our algorithm is based more generally on the mirror map, which constitutes a generalization of the Euclidean projection. For this reason, we are not able to use the usual fixed-point approach for variational inequality. Moreover, the asymmetric algorithm proposed in \cite{Paccagnan2017} uses constant step size in contrast to our algorithm, which uses variable step size.

We take the inspiration for the method of augmentation of the Lagrangian from the work \cite{Mahdavi1}. There, the authors provide an algorithm for online optimization with sub-linear regret bound able to handle constraint. The method of augmentation of the Lagrangian helps to obtain a sub-linear bound on the violation of constraints.

\subsubsection*{Basic Notations}
In this work we consider always the linear Euclidean space $\real^{D}$. The projection onto the closed convex subset $A$ of $\real^{D}$ is denoted by $\Pi_{A}$. The dual norm of a norm $\norm{\cdot}$ on $\real^{D}$ is denoted by $\norm{\cdot}_{*}$. $F:\real^{D}\rightarrow\real^{D}$ is said to be Lipschitz continuous on an a non-empty subset $\Z\subset (\real^{D},\norm{\cdot})$ with constant $L>0$ if $$\norm{F(x)-F(z)}_{*}\leq L\norm{x-z},\quad\forall x,z\in\mathcal{Z}$$. $F$ is said to be monotone on $\Z$ if: $$\inn{x_{1}-x_{2}}{F(x_{1})-F(x_{2})}\leq 0,$$
for all $x_{1},x_{2}\in \mathcal{Z}$. If in the latter strict inequality hold for $x_{1}\neq x_{2}$, then $F$ is said to be strictly monotone.

\section{Model Descriptions}
We consider a non-cooperative game (NG) $\Gamma$ played by a finite set of players $[N]=\lrbrace{1,\ldots,N}$. During the game each player $i\in[N]$ can choose an action/strategy $x^{(i)}$ from a non-empty compact convex subset $\X_{i}$ of a finite-dimensional normed space $(\real^{D_{i}},\norm{\cdot}_{i})$. 
A usual assumption on the action set is the following:
\begin{assum}
$\X_{i}$ is a non-empty compact convex subset of a finite dimensional space $\mathcal{V}_{i}$.
\end{assum} 
The payoff/reward for player $i\in [N]$ is given by the function $u_{i}:\X\rightarrow\real$, where $\X=\prod_{i}\X_{i}$, and the actual action/strategy-profile $x=(x^{(1)},\ldots,x^{(N)})\in\X:=\prod_{i}\X_{i}$. If we work with the whole population, we consider the normed space $(\prod_{i=1}^{N}\real^{D_{i}},\normi{\cdot})$, where $\normi{x}:=\sum_{i}\norm{x^{(i)}}_{i}$. In order to highlight the action of player $i$ we often write $x=(x^{(i)},x^{(-i)})$ where $\mathbf{x}^{(-i)}=(x^{(j)})_{j\neq i}$. Moreover, we mostly assume in this work the following regularity condition for the utility functions:
\begin{assum} 
\label{Ass:aiishhfggfjhhdjjdhhdjjddd}
For all $i\in [N]$ and $x^{(-i)}\in \X_{-i}$,
 $u_{i}((\cdot),x^{(-i)})$ is 
      concave  and $$v_{i}(x):=\nabla_{x^{(i)}}u_{i}(x)$$ is continuous. 
\end{assum}

In this work, we are specifically interested in NG $\Gamma$ with coupled constraints (NGCC), i.e. in NG $\Gamma$ which is in addition subject to coupled inequality constraints $$\mathcal{C}:= \{x\in\real^{N}:~g_{i}(x)\leq 0,~ i\in [M]\},$$ or in the vectorized form $\mathcal{C}:=\{g(x)\leq 0\}$. So the set of feasible strategy profile is $$\mathcal{Q}:=\mathcal{C}\cap\mathcal{X}$$ which we assumed to be non-empty, and correspondingly the set of feasible strategy for player $i$ is $$\mathcal{Q}^{(i)}(x^{(-i)}):=\{x^{(i)}\in\mathcal{X}_{i}:~g(x)\leq 0,~\}.$$ We denote NGCC by $\Gamma=([N],u,\X,\mathcal{C})$, where $u:=(u_{1},\ldots,u_{N})$.
For simplicity, we consider more specifically linear constraint, where $$g(x)=Ax-b ~\text{with}~ A=[A_{(:,1)},\ldots,A_{(:,N)}]\in\real^{M\times \sum_{i=1}^{N}D_{i}},$$ where $A_{(:,j)}\in \real^{M\times D_{j}}$, and where $b\in\real^{M}$. In this case, each agents is assume to now only its contribution to the inequality constraints, which means that $A_{(:,i)}$ is only visible to agent $i$. 

The following regularity condition on $\mathcal{Q}$ is useful for later purposes:
\begin{assum}[Slater's condition]
\label{Eq:aioaoosjjshhdjjddss}
There exists $x_{*}\in\text{relint}{(\X)}$ s.t. $Ax_{*}< b$, where $\text{relint}(\X)$ denotes the relative interior of $\X$.
\end{assum}

\subsubsection*{Mirror Map and Fenchel Coupling}
In general, first order evolution takes place in the dual space. So, in order to realize their actions, the agents need a mapping to project the iterate back to their individual constraint sets. A canonical way to do this is by means of the following:
\begin{definition}[Regularizer/penalty fct. and Mirror Map]
Let $\Z$ be a compact convex subset of a normed space $(E,\norm{\cdot})$, and $K>0$. We say $\psi:\Z\rightarrow\mathbb{R}$ is a $K$-strongly convex \textit{regularizer} (or \textit{penalty function}) on $\Z$, if $\psi$ is continuous and $K$-strongly convex on $\Z$, in the sense that for all $x,y\in \mathcal{X}$ and $\lambda\in [0,1]$:
\begin{equation*}
\psi(\lambda x+(1-\lambda)y)\leq \lambda \psi(x)+(1-\lambda)\psi(y)-\tfrac{K}{2}\lambda (1-\lambda)\norm{x-y}^{2}.
\end{equation*}
The mirror map $\Phi:E^{*}\rightarrow\Z$ induced by $\psi$ is defined by:
\begin{equation*}
\Phi(y):=\argmax\limits_{x\in\mathcal{X}}\left\{\left\langle x,y\right\rangle-\psi(x)\right\}.
\end{equation*}
\end{definition}
In case that the convex conjugate: $$\psi^{*}(y)=\max_{x\in\Z}\left\{\left\langle y,x\right\rangle-\psi(x)\right\}$$ of the $K$-strongly convex regularizer $\psi$ on $\mathcal{Z}$ is known, one can compute the mirror map by $\Phi=\nabla \psi^{*}$. Moreover, it can be shown that $\Phi$ is $1/K$-Lipschitz continuous. 
For a proof of those facts, see e.g. Theorem 23.5 in \cite{Rockafellar1970} and Theorem 12.60(b) in \cite{Rockafellar1998}. 

Mirror map constitutes a generalization of the usual Euclidean projection operator. Interesting example of mirror map is the so called logit choice:
 $$\Phi(y)=\frac{\exp(y)}{\sum_{l=1}^{D}\exp(y_{l})}$$ which is generated by the penalty function: $$\psi(x)=\sum_{k=1}^{D}x_{k}\log x_{k},$$ known as the Gibbs entropy, on the simplex $\Delta\subset(\real^{D},\norm{}_{1})$.

As noticed in \cite{Mertikopoulos2016}, a convex regularizer induces canonically the following notion of "distance":
\begin{definition}[Fenchel Coupling \cite{Mertikopoulos2016}]
Let $\psi:\X\rightarrow\real$ be a penalty function on $\X$. Then the Fenchel coupling induced by $\psi$ is defined as
\begin{equation*}
F(p,y)=\psi(p)+\psi^{*}(y)-\inn{p}{y},\quad p\in \X,~y\in E^{*}.
\end{equation*}
\end{definition}
Some useful properties of the Fenchel coupling is stated in the following (for proof see \cite{Mertikopoulos2016}):
\begin{proposition}
\label{Prop:aaisshhfjffjfjfff}
Let $F$ be the Fenchel coupling induced by a $K$-strongly convex regularizer of $\X$. For $p\in\mathcal{X}$, $y,y^{'}\in \mathcal{V}^{*}$, we have:
\begin{enumerate}
\item $F(p,y)\geq (K/2)\norm{\Phi(y)-p}^{2}$
\item $F(p,y^{'})\leq F(p,y)+\inn{y^{'}-y}{\Phi(y)-p}+(1/2K)\norm{y^{'}-y}^{2}_{*}$
\end{enumerate}
\end{proposition}
 
Throughout this work, we assume that each agent $i\in [N]$ possess a $K_{i}$-strongly convex regularizer $\psi_{i}$ which induces the mirror map $\Phi_{i}$, and the Fenchel coupling $F_{i}$. In order to emphasize the action of the whole population, we sometimes use the operator $\mathbf{\Phi}:\prod_{i}\real^{D_{i}}\rightarrow\X$, $y\mapsto (\Phi_{1}(y^{(1)}),\ldots,\Phi_{N}(y^{(N)}))$ and the "total" Fenchel coupling $F^{N}:\X\times\prod_{i}\real^{D_{i}}$, $(x,y)\mapsto\sum_{i}F_{i}(x_{i},y_{i})$.

\subsubsection*{Algorithm}

The evolution of the agents which we concerned with in this work is given in the following:
\begin{algorithm}
\label{Alg:aoaishhjddhhdddd}
\caption{Mirror ascent with Augmented Lagrangian (MAAL)}
\begin{algorithmic}
\REQUIRE Step size sequence $(\gamma_{t})_{t}$, augmentation sequence $(\theta_{t})_{t}$, initial dual action $Y_{0}^{(i)}\in\V^{*}_{i}$, and initial dual variable $\lambda_{0}$.
\FOR{$t=1,2,\ldots$}
\FOR{every player $i\in [N]$ }
\STATE Play $X^{(i)}_{t}\leftarrow \Phi_{i}(Y^{(i)}_t)$
\STATE Observe $v_{i}(X_{t})$
\STATE Update $Y^{(i)}_{t+1}\leftarrow Y_{t}^{(i)}+\gamma_{t}(v_{i}(X_{t})-A^{T}_{(:,i)}\lambda_{t})$
\ENDFOR

\STATE Central operator update:
   \begin{equation*}
	 \lambda_{t+1}\leftarrow\Pi_{\real^{M}_{\geq 0}}\left(\lambda_{t}+\gamma_{t}\left[\left(A X_{t}-b\right)-\theta_{t}\lambda_{t}\right]\right)
	 \end{equation*}
\STATE Central operator broadcast $\lambda_{t+1}$ to all players. 
\ENDFOR
\end{algorithmic}
\end{algorithm}

The difference between usual online mirror ascent is that the gradient update for agent $i$ has an additional term involving his contribution to the constraint set ($A^{T}_{(:,i)}$) and the price vector $\lambda_{t}$ provided by the central coordinator. For this reason, we speak of semi-decentralized update. The price vector is updated via projected gradient ascent for maximizing the augmented Lagrangian dual objective. The aspect of augmentation of Lagrangian is reflected in the term $-\theta_{t}\lambda_{t}$. By this reason, $(\theta_{t})_{t}$ is called augmentation sequence.

\section{Variational Description of Equilibriums}
\label{Sec:aoaosjjskdjdjdkkdd}
A classical notion of equilibrium is the Nash equilibrium. It describes the state in which no agent can increase his payoff by unilaterally changing his strategy:
\begin{definition}[Nash Equilibrium]
$x_{*}\in\mathcal{C}$ is a Nash equilibrium of the NGCC $\Gamma=([N],u,\X,\mathcal{C})$, if for every $i\in [N]$:
\begin{equation}
\label{Eq:aiaishhdjhdggshhggaa}
u_{i}(x_{*})\geq u_{i}(x^{(i)},x_{*}^{(-i)}),\quad\forall x^{(i)}\in\mathcal{Q}^{(i)}(x^{(-i)}_{*})
\end{equation}
\end{definition}
\subsection{Variational Inequality and Nash equilibrium}
Rather than with the concept of Nash equilibrium, it is advantageous from the analytical point of view to work with the concept of the so-called variational inequality (VI):
\begin{definition}
Let $\Z$ be a subset of a finite dimensional normed space $(E,\norm{\cdot})$, and suppose that $F:\mathcal{Z}\rightarrow E^{*}$. A point $\overline{x}\in\Z$ is a solution of the variational inequality $\text{VI}(\mathcal{Z},F)$, if $\left\langle x-\overline{x} ,F(\overline{x})\right\rangle\leq 0$, $\forall x\in\mathcal{Z}$. The set of solution of $\VI(\X,F)$ is denoted by $\SOL(\X,F)$.
\end{definition}
The usual first order optimality condition for convex optimization asserts the following relation between two concepts:
\begin{proposition}
\label{Prop:aaoaojssjhhdjdhsss}
If Assumption \ref{Ass:aiishhfggfjhhdjjdhhdjjddd} holds, then $\SOL (\mathcal{Q},v)$ is a subset of the set of Nash equilibriums.
\end{proposition}
In the case where no coupling constraint is present, i.e. $\mathcal{C}=\mathcal{X}$, then the converse of above proposition holds. However, due to the coupling constraint, a Nash equilibrium has not to be a solution of variational inequality. 

Another nice thing about VI is that under mild condition one can establish existence of its solution. For instance it is known that in case $\Z\neq\emptyset$ is compact and convex and $F$ is continuous, then there exists at least a solution of $\SOL(\Z,F)$. Moreover in case that $F=v$ and $\Z=\mathcal{Q}$, the latter and Proposition \ref{Prop:aaoaojssjhhdjdhsss} implies the existence of a Nash equilibrium for $\Gamma$:   
\begin{proposition}
Suppose that the Assumption \ref{Ass:aiishhfggfjhhdjjdhhdjjddd} holds. Then $\Gamma$ has a Nash equilibrium. In case that Assumption \ref{Ass:aiishhfggfjhhdjjdhhdjjddd} holds in the strict manner, then $\Gamma$ has a unique Nash equilibrium.
\end{proposition}
\subsection{Decoupling the Constraints by means of Lagrangian Method}
As we have already seen, the equilibrium of the constrained game $\Gamma$ is related to the solution of the variational inequality $\VI(\mathcal{Q},v)$. In order to analyze $\VI(\mathcal{Q},v)$ it is convenient to extend the previous problem to the problem $\VI(\mathcal{X}\times \real^{M}_{+},\tilde{v})$, where $\tilde{v}:\mathcal{X}\times \real^{M}_{+}$,
\begin{equation*}
\tilde{v}:\mathcal{X}\times \real^{M}_{+},\quad (x,\lambda)\mapsto \left[v(x)-A^{T}\lambda, Ax - b\right]^{T}.
\end{equation*}
The advantage of this method is the decoupling of the constraint set, i.e. we only have to handle with the constraint set $\X\times \real^{M}_{\geq 0}$ with product structure rather than with $\mathcal{Q}$. The following shows that there is no burden in doing this:
\begin{proposition}
\label{Prop:aiaisshshjdhdggddd}
Suppose that Assumption \ref{Ass:aiishhfggfjhhdjjdhhdjjddd} and Assumption \ref{Eq:aioaoosjjshhdjjddss} holds. The following statements are equivalent:
\begin{enumerate}
\item $\overline{x}\in\mathcal{Q}$ is a solution of $\VI(\mathcal{Q},v)$
\item There exists $\overline{\lambda}\in\real^{M}_{\geq 0}$ s.t. $(\overline{x},\overline{\lambda})$ is a solution of $\VI(\mathcal{X}\times \real^{M}_{+},\tilde{v})$. 
\end{enumerate}
\end{proposition}
The proof is standard KKT argumentation and based on e.g. 1.3.4 Proposition in \cite{Facchinei1} (see also Subsubsection 4.3.2.2 in \cite{Scutari2012}).

So in order to solve $\VI(\mathcal{Q},v)$ it is sufficient to seek for the solution of $\VI(\mathcal{X}\times \real^{M}_{+},\tilde{v})$. It follows from Proposition \ref{Prop:aoaoajsjsksjjdddddd}, we need to seek for the latter for variationally stable set for $(\mathcal{X}\times \real^{M}_{+},\tilde{v})$, assuming that it is non-empty. 
\subsection{Variational Inequality and Variational Stability}
\label{Subsec:aooaajajsskkjjddddff}
In case that $G$ is a monotone operator, it holds:
\begin{equation*}
\inn{x-\overline{x}}{G(x)}\leq\inn{x-\overline{x}}{G(\overline{x})}\leq 0,~\forall x\in\mathcal{Z}, \overline{x}\in\text{SOL}(\mathcal{Z},G).
\end{equation*}
This motivates to introduce the following notion:
\begin{definition}
We say that a closed $\VS(\Z,F)\subset\Z$ is a variationally stable set for $(\Z,F)$, if:
\begin{equation}
\label{Eq:aoaosjsjhdhdjsjsjsdd}
\inn{x-\overline{x}}{F(x)}\leq 0\quad\forall x\in \Z,~\overline{x}\in\VS(\Z,F)
\end{equation}
with equality for a given $\overline{x}\in\VS(\Z,F)$ if and only if $x\in \VS(\Z,F)$.
\end{definition}

As we will see later, it is convenient algorithmically to work with the concept of variational stability instead with the concept of variational inequality. However, the following gives that under mild condition, both concepts are the same:
\begin{proposition}
\label{Prop:aoaoajsjsksjjdddddd}
Suppose that $\VS(\Z,F)\neq\emptyset$. Then $\VS(\Z,F)=\SOL(\Z,F)$.
\end{proposition}
\begin{proof}
Let $\overline{x}\in\SOL(\Z,F)$ but $\overline{x}\notin \VS(\Z,F)$. We have since $\overline{x}\in\SOL(\Z,F)$, $\inn{x-\overline{x}}{F(\overline{x})}\leq 0$, $\forall x\in \Z$. So in particular for an $x_{*}\in\tilde{\mathcal{Z}}$:
\begin{equation}
\label{Eq:aioaaosjsjdhhdfggffgdhhdd}
\inn{x_{*}-\overline{x}}{F(\overline{x})}\leq 0.
\end{equation}
Moreover, since $\overline{x}\notin\VS(\Z,F)$ and $x_{*}\in\VS(\Z,F)$, we have $\inn{\overline{x}-x_{*}}{F(\overline{x})}<0$ and thus $\inn{x^{*}-x}{F(\overline{x})}>0$. This contradicts to \eqref{Eq:aioaaosjsjdhhdfggffgdhhdd}, so that we can imply the desired statement. 
\end{proof}
\begin{remark}
The assumption that $\VS(\mathcal{Z},F)\neq\emptyset$ is closed appears at the first sight forced. However, if $\mathcal{Z}$ is a nonempty compact convex set, $F=\nabla g$ where $g:\mathcal{Z}\rightarrow \real$ is a concave function then the assumption is true. Indeed Since $g$ is concave, $\nabla g$ is monotone. Therefore:
\begin{equation*}
\inn{x-\overline{x}}{\nabla g(x)-\nabla g(\overline{x})}\leq 0.
\end{equation*}  
By the first-order optimality condition we have for $x^{*}\in \argmax g$ and $x\in\mathcal{Z}$, $\inn{\nabla g(\overline{x})}{x-\overline{x}}\leq 0$ and thus by monotonicity:
\begin{equation*}
\inn{x-\overline{x}}{\nabla g(x)}\leq \inn{x-\overline{x}}{\nabla g(\overline{x})}\leq 0.
\end{equation*}
Moreover concavity asserts that $\inn{\nabla g(x)}{x-x^{*}}<0$ whenever $x$ is not a maximizer of $g$. Thus we have that $\argmax g=\VS(\mathcal{Z},\nabla g)$ and the fact that $\argmax g\neq \emptyset$ implies that $\VS(\mathcal{Z},\nabla g)\neq \emptyset$. Moreover it is easy to see that $\argmax g$ is closed.
\end{remark}\section{Bound for Primal-Dual Iterate via Fenchel Coupling}
We begin by measuring the distance between the evolution of each agents and a strategy profile by means of the "total" Fenchel coupling $F^{N}$, which is crucial to provide convergence theorem for MAAL. By using Proposition \ref{Prop:aaisshhfjffjfjfff}, inserting the iterate of the algorithm, using triangle inequality, we have for all $x\in\X$:
\begin{align*}
F^{N}(x,Y_{t+1})-F^{N}(x,Y_{t})\leq&\gamma_{t}\inni{X_{t}-x}{v(X_{t})-A^{T}\lambda_{t}}\\
&+\tfrac{\gamma_{t}^{2}}{2K}(C_{1}^{2}+C_{2}^{2}\norm{\lambda_{t}}^{2}_{2})
\end{align*}

where $C_{1},C_{2}>0$ are constants fulfilling:
\begin{equation}
\label{Eq:aoaoshshjhsjjss}
\normi{v(x)}_{*}\leq C_{1},~\normi{A^{T}\lambda}_{*}\leq C_{2}\norm{\lambda}_{2},~\forall x\in\mathcal{X},\lambda\in\real^{M}_{\geq 0}.
\end{equation}
By summing over all $t=0,\ldots,T$ and subsequent telescoping, we obtain a bound for $$\mathcal{E}_{T}^{(1)}(x):=F^{N}(x,Y_{T})-F^{N}(x,Y_{0}).$$
That is:
\begin{align}
\mathcal{E}_{T}^{(1)}(x)\leq&\sum_{t=0}^{T}\gamma_{t}\left[\inni{X_{t}-x}{v(X_{t})-A^{T}\lambda_{t}}\right]+\tfrac{C_{1}^2}{2K}\sum_{t=0}^{T}\gamma_{t}^{2}\nonumber\\
&+\tfrac{C_{2}^{2}}{2K}\sum_{t=0}^{T}\gamma_{t}^{2}\norm{\lambda_{t}}^{2}_{2},\label{Eq:aiashhshjdhfggfsaas}
\end{align}

In order to eliminate the term \eqref{Eq:aiashhshjdhfggfsaas} involving the dual iterate $\lambda_{t}$, we now estimate of the distance between the dual iterate and any dual point. We can bound:
 $$\mathcal{E}_{T}^{(2)}(\lambda):=(\norm{\lambda-\lambda_{T}}^2_{2}-\norm{\lambda-\lambda_{0}}_{2}^{2})/2$$ for any $\lambda\in\real^{M}_{\geq 0}$ by:
\begin{align}
\mathcal{E}_{T}^{(2)}(\lambda)\leq &\sum_{t=0}^{T}\gamma_{t}\inn{\lambda_{t}-\lambda}{A X_{t}-b}-\sum_{t=0}^{T}\tfrac{\gamma_{t}\theta_{t}}{2}(\norm{\lambda_{t}}^{2}_{2}-\norm{\lambda}^{2})\nonumber\\
&+\sum_{t=0}^{T}\gamma_{t}^{2}(C_{3}^{2}+\theta_{t}^{2}\norm{\lambda_{t}}^2_{2}),\label{Eq:aiaisoodhhfhffgfgfhhfff}
\end{align}
where $C_{3}>0$ is a constant fulfilling:
\begin{equation}
\label{Eq:aoaosjsjhfjfggdhdhgddf}
\norm{A x}_{2}\leq C_{3}.
\end{equation}
\eqref{Eq:aiaisoodhhfhffgfgfhhfff} can be proven in the similar manner as the proof of \eqref{Eq:aiashhshjdhfggfsaas}. 

By combining \eqref{Eq:aiashhshjdhfggfsaas} and \eqref{Eq:aiaisoodhhfhffgfgfhhfff} we obtain immediately the following estimate for the evolution of: $$\tilde{F}((x,\lambda),(Y_{T},\lambda_{T})):=F^{N}(x,Y_{T})+(\norm{\lambda_{t}-\lambda}^{2}_{2}/2),$$ 
which is:
\begin{theorem}
\label{Thm:aoaaojsjksjddjjdjddd}
Let $C_{1},C_{2},C_{3}>0$ be constants fulfilling \eqref{Eq:aoaoshshjhsjjss} and \eqref{Eq:aoaosjsjhfjfggdhdhgddf}. It holds for:
\begin{equation*}
\mathcal{E}_{T}(x,\lambda):=\tilde{F}((x,\lambda),(Y_{T},\lambda_{T}))-\tilde{F}((x,\lambda),(Y_{0},\lambda_{0}))
\end{equation*}
and for all $(x,\lambda)\in\mathcal{X}\times \real^{M}_{\geq 0}$:
\begin{align}
\mathcal{E}_{T}(x,\lambda)&\leq\sum_{t=0}^{T}\gamma_{t}\inni{(X_{t},\lambda_{t})-(x,\lambda)}{\tilde{v}(X_{t},\lambda_{t})}_{\sim}\nonumber\\
&+\tilde{C}_{1}\sum_{t=0}^{T}\gamma_{t}^{2}+\sum_{t=0}^{T}\tfrac{\gamma_{t}\theta_{t}\norm{\lambda}^{2}_{2}}{2}\nonumber\\&+\sum_{t=0}^{T}\gamma_{t}\norm{\lambda_{t}}^{2}_{2}\left[\gamma_{t}\left(2\theta_{t}^{2}+\tilde{C}_{2}\right)-\tfrac{\theta_{t}}{2}\right]\nonumber
\end{align}
where:
\begin{equation*}
\tilde{C}_{1}:=\tfrac{C_{1}^2}{2K}+2 C_{3}^{2}\quad\tilde{C}_{2}:=\tfrac{C_{2}^{2}}{2K},
\end{equation*}
and for all $x,\tilde{x}\in \prod_{i}\real^{D_{i}}$ and $\lambda,\tilde{\lambda}\in\real^{M}$:
\begin{equation*}
\inni{(x,\lambda)}{(\tilde{x},\tilde{\lambda})}_{\sim}=\inni{x}{\tilde{x}}+\inn{\lambda}{\tilde{\lambda}}.
\end{equation*}
\end{theorem}
\section{Convergence Analysis}
In this section we investigate the convergence of MAAL to the variational stable set $\VS(\mathcal{Q},v)$. As already discuss in Section \ref{Sec:aoaosjjskdjdjdkkdd}, this leads, in the case that $\VS(\mathcal{Q},v)\neq \emptyset$, to the convergence of MAAL to the solution $\SOL(\mathcal{Q},v)$ of variational inequality $\VI(\mathcal{Q},v)$ and to the convergence of MAAL to the corresponding subset of the Nash equilibrium of $\Gamma$. By $1.$ in proposition \ref{Prop:aaisshhfjffjfjfff} it follows that convergence with respect to $F^{N}$ implies the convergence of the iterate w.r.t. the underlying norm $\normi{\cdot}$. Therefore the bound for $F^{N}(x,Y_{T})$ provided in previous section helps us to establish the desired statement. For technical reason, it is advantageous to have the converse property:
\begin{assum}
\label{Ass:aoasjsdhhggehegeeeerrrttt}
For any $p\in\X$ and any sequence $(Y_{n})_{n}$ in $\V^{*}$, it holds: $\Phi(Y_{n})\rightarrow p~\Rightarrow~ F^{N}(p,Y_{n})\rightarrow 0$
\end{assum}
Define for $\mathcal{C}\subset \X$ and $\tilde{\mathcal{C}}\subset \X\times\real^{M}_{\geq 0}$:
\begin{align*}
F^{N}(\mathcal{C},y)&:=\inf\lrbrace{F^{N}(x,y):~x\in\mathcal{C}}\\
\tilde{F}(\tilde{\mathcal{C}},z)&:=\inf\lrbrace{\tilde{F}(x,y):~x\in\mathcal{C}}
\end{align*}
Notice that the property given \eqref{Ass:aoasjsdhhggehegeeeerrrttt} holds also in case that $p$ is substituted more generally by a closed set:

\begin{proposition}
\label{Prop:aoaojsjsksjsjsjkaakaaks}
Suppose that Assumption \ref{Ass:aoasjsdhhggehegeeeerrrttt} holds. Let $\mathcal{C}$ be a closed subset of $\X$ and $\tilde{\mathcal{C}}$ be a closed subset of $\X\times \real^{M}_{\geq 0}$. Then $\Phi(Y_{t})\rightarrow \mathcal{C}$ if and only if $F^{N}(\mathcal{C},y)\rightarrow 0$ and $(\Phi(Y_{t}),\tilde{\lambda}_{t})\rightarrow \tilde{\mathcal{C}}$ if and only if $F^{N}(\tilde{\mathcal{C}},(\Phi(Y_{t}),\tilde{\lambda}_{t}))\rightarrow 0$ 
\end{proposition}

In the following, we state the following convergence statement for the iterate $(X_{t},\lambda_{t})$ of MAAL:
\begin{theorem}
\label{Lem:aoaosjsjshdhdjshssjjddd}
Let $\tilde{C}_{2}>0$ be a constant as given in Theorem \ref{Thm:aoaaojsjksjddjjdjddd}. Suppose that Assumption \ref{Ass:aoasjsdhhggehegeeeerrrttt} holds. Suppose that $(\gamma_{t})_{t}$ satisfies:
\begin{equation}
\label{Eq:aiiahhssjjdhhdjjdjdsss}
\sum_{t=0}^{\infty}\gamma_{t}=\infty,\quad \tfrac{\sum_{t=1}^{T}\gamma^{2}_{t}}{\sum_{t=1}^{T}\gamma_{t}}\rightarrow 0,~T\rightarrow\infty.
\end{equation}
For an augmentation sequence $(\theta_{t})_{t}$ satisfying:
\begin{equation}
\label{Eq:aoaoshshsjjhdggdhhss}
\tfrac{\sum_{t=1}^{T}\gamma_{t}\theta_{t}}{\sum_{t=1}^{T}\gamma_{t}}\rightarrow 0,\quad T\rightarrow\infty,
\end{equation} 
and:
\begin{equation}
\label{Eq:aooajsjsjskkjdjjjddd}
\gamma_{t}\left(2\theta_{t}^{2}+\tilde{C}_{2}\right)-\tfrac{\theta_{t}}{2}\leq 0,~\text{for large }t\geq 0.
\end{equation}
It holds for the iterates of MAAL:
\begin{enumerate}
\item There exists a subsequence $(X_{t_{k}},\lambda_{t_{k}})_{k}$ of $(X_{t},\lambda_{t})_{t}$ s.t. $(X_{t_{k}},\lambda_{t_{k}})\rightarrow\VS(\X\times\real^{M}_{\geq 0},\tilde{v})$ as $k\rightarrow\infty$.
\item $(X_{t},\lambda_{t})\rightarrow \VS(\X\times\real^{M}_{\geq 0},\tilde{v})$ as $t\rightarrow\infty$,
\end{enumerate}
\end{theorem}
\begin{proof}
To show the first statement of the Theorem, notice that:
\begin{align}
&\mathcal{E}_{T}(x_{*},\lambda_{*})\nonumber\\
&\leq\tau_{T}\left(\tfrac{\sum_{t=0}^{T}\gamma_{t}\xi_{t}(x_{*},\lambda_{*})}{\tau_{T}}+\tilde{C}_{1}\tfrac{\sum_{t=0}^{T}\gamma^{2}_{t}}{\tau_{T}}+\tfrac{\sum_{t=0}^{t}\gamma_{T}\psi_{t}}{\tau_{T}}\right),\label{Eq:aiaishhfggfhfhffff}
\end{align}
where $\tau_{T}:=\sum_{t=0}^{T}\gamma_{k}$,
\begin{align}
&\xi_{t}(x_{*},\lambda_{*}):=\inni{(X_{t},\lambda_{t})-(x_{*},\lambda_{*})}{\tilde{v}(X_{t},\lambda_{t})}_{\sim},\nonumber\\
\psi_{t}&:=\norm{\lambda_{t}}^{2}_{2}\left[\gamma_{t}\left(2\theta_{t}^{2}+\tilde{C}_{2}\right)-\tfrac{\theta_{t}}{2}\right]+\tfrac{\theta_{t}\norm{\lambda}^{2}_{2}}{2}\leq\tfrac{\theta_{t}\norm{\lambda}^{2}_{2}}{2},\label{Eq:aoaouudzzdzehehhee}
\end{align}
where the inequality in \eqref{Eq:aoaouudzzdzehehhee} follows by \eqref{Eq:aooajsjsjskkjdjjjddd}. Let be $U$ be an arbitrary neighborhood (w.r.t. a norm e.g. $\norm{\cdot}_{2}$) of $\VS(\X\times\real_{\geq 0}^{M},\tilde{v})$. Suppose that $(X_{t},\lambda_{t})\notin U$ for all sufficiently large $t\geq 0$. We may assume w.l.o.g. that $X_{t}\notin U$ for all $t\geq 0$. So for all $(x_{*},\lambda_{*})\in\VS(\mathcal{X}\times\real^{M}_{\geq 0},v)$, it follows that we can find $c>0$ s.t. $\xi_{t}(x_{*},\lambda_{*})\leq -c$, $\forall t\geq 0$.
This yields:
\begin{equation}
\label{Eq:aoaojdhhddjeeerrr}
\mathcal{E}_{T}(x_{*},\lambda_{*})\leq\tau_{T}\left(-c+\tilde{C}_{1}\tfrac{\sum_{t=0}^{T}\gamma^{2}_{t}}{\tau_{T}}+\tfrac{\sum_{t=0}^{t}\gamma_{t}\psi_{t}}{\tau_{T}}\right).
\end{equation}
\eqref{Eq:aiiahhssjjdhhdjjdjdsss} (resp. \eqref{Eq:aoaouudzzdzehehhee} and \eqref{Eq:aoaoshshsjjhdggdhhss}) gives that the second (resp. third) summand in \eqref{Eq:aoaojdhhddjeeerrr} converges to 0 as $t$ goes to infinity.
Finally, by the fact that $\tau_{T}\rightarrow\infty$ as $T\rightarrow\infty$, we have $\mathcal{E}_{T}(x_{*},\lambda_{*})\rightarrow-\infty,\quad n\rightarrow\infty$, which contradicts the fact that $\tilde{F}((x_{*},\lambda_{*}),(Y_{0},\lambda_{0}))$ is finite. Thus $(X_{t},\lambda_{t})\in U$ for infinitely many $t\geq 0$. 

To show the convergence of $(X_{t},\lambda_{t})$, i.e. the second statement of the Theorem, it is sufficient by $2.$ in Proposition \ref{Prop:aaisshhfjffjfjfff} to show that for all $\epsilon >0$, and:
\begin{equation*}
U_{\epsilon}:=\lrbrace{(x,\lambda):x=\Phi(y),\tilde{F}(\VS(\mathcal{X}\times \real^{M}_{\geq 0},\tilde{v}),(y,\lambda))<\epsilon}
\end{equation*}
$(X_{t},\lambda_{t})\in U_{\epsilon}$ for all but finite $t\in\nat$. Toward this end, we show that for sufficiently large $t$, $(X_{t},\lambda_{t})\in U_{\epsilon}$ implies $(X_{t+1},\lambda_{t+1})\in U_{\epsilon}$. Combining this fact with Lemma \ref{Lem:aoaosjsjshdhdjshssjjddd} yields finally the desired statement.

We have by $3.$ in Proposition \ref{Prop:aaisshhfjffjfjfff} and by $1.$ in Lemma \ref{Lem:aoaosjsjshdhdjshssjjddd}:
\begin{align}
&\tilde{F}((x_{*},\lambda_{*}),(Y_{t+1},\lambda_{t+1}))
\leq  \tilde{F}((x_{*},\lambda_{*}),(Y_{t+1},\lambda_{t+1}))\nonumber\\
&+\gamma_{t}\xi_{t}((x_{*},\lambda_{*}))+\gamma_{t}\psi_{t}+\gamma_{t}^{2}\tilde{C}_{1},\label{Eq:aoaosjsjsdjdjdhfgfhgfhdhdgfgf}
\end{align}
for a constant $\tilde{C}>0$. Suppose that $(X_{t},\lambda_{t})\in U_{\epsilon}$. By Assumption \ref{Ass:aoasjsdhhggehegeeeerrrttt} it holds that $U_{\epsilon/2}$ contains a neighborhood of $\VS(\mathcal{X}\times\real^{M}_{\geq 0},\tilde{v})$ (say w.r.t. $\norm{\cdot}_{2}$). Otherwise we have that any neighborhood of $\VS(\mathcal{X}\times\real^{M}_{\geq 0},\tilde{v})$ w.r.t. $\normi{\cdot}$ is not contained $U_{\epsilon/2}$, and since the image of $\Phi$ coincides with the domain of the subdifferential of $\psi$ and $\psi$ is subdifferentiable on the interior of $\X$, we can choose a sequence $(\tilde{Y}_{t})_{t}$ in $\mathcal{V}^{*}$ and a sequence $(\tilde{\lambda}_{n
t})$ in $\real^{M}_{\geq 0}$ satisfying $(\Phi(\tilde{Y}_{t}),\tilde{\lambda}_{t})\rightarrow \VS(\mathcal{X}\times\real^{M}_{\geq 0},\tilde{v})$ but $(\Phi(\tilde{Y}_{t}),\tilde{\lambda}_{t})\notin U_{\epsilon/2}$, i.e. $\tilde{F}(\VS(\mathcal{X}\times\real^{M}_{\geq 0},\tilde{v}),(\Phi(\tilde{Y}_{t}),\tilde{\lambda}_{t}))\geq \epsilon/2$. 
Since $\VS(\mathcal{X}\times\real^{M}_{\geq 0},\tilde{v})$ is closed, the latter contradicts with Proposition \ref{Prop:aoaojsjsksjsjsjkaakaaks}. 

Now, an implication of the fact that $U_{\epsilon/2}$ contains a neighborhood of $\VS(\mathcal{X}\times\real^{M}_{\geq 0},\tilde{v})$ is that for all $(x_{*},\lambda_{*})\in\VS(\mathcal{X}\times\real^{M}_{\geq 0},\tilde{v})$, there exists $c>0$ s.t.
\begin{equation*}
\inni{\tilde{v}(x,\lambda)}{(x,\lambda)-(x_{*},\lambda_{*})}\leq -c,\quad \forall x\in U_{\epsilon}\setminus U_{\epsilon/2}
\end{equation*}
So, if $X_{n}\in U_{\epsilon}\setminus U_{\epsilon/2}$ \eqref{Eq:aoaosjsjsdjdjdhfgfhgfhdhdgfgf} yields:
\begin{align*}
&\tilde{F}((x_{*},\lambda_{*}),(Y_{t+1},\lambda_{t+1}))\\
&\leq  \tilde{F}((x_{*},\lambda_{*}),(Y_{t},\lambda_{t}))+\gamma_{t}(-c+\psi_{t}+\gamma_{t}\tilde{C}_{1}),
\end{align*}
By \eqref{Eq:aoaouudzzdzehehhee}, we have $\abs{\psi_{t}}\rightarrow 0$ for $t\rightarrow\infty$ and thus for large enough $t\in\nat$, there exists $\tilde{c}>0$ s.t. $\tilde{F}((x_{*},\lambda_{*}),(Y_{t+1},\lambda_{t+1}))\leq \tilde{F}((x_{*},\lambda_{*}),(Y_{t},\lambda_{t}))+\gamma_{t}(-\tilde{c}+\gamma_{t}\tilde{C}_{1})$. 
and consequently for sufficiently large $t\in\nat$ s.t. the inequality before and $\gamma_{t}\leq \tilde{c}/\tilde{C}$ holds, we have $\tilde{F}(\VS(\X\times\real^{M}_{\geq 0}),(Y_{t+1},\lambda_{t+1}))\leq  \tilde{F}(\VS(\X\times\real^{M}_{\geq 0}),(Y_{t},\lambda_{t})))<\epsilon$,
since by assumption $(X_{t},\lambda_{t})\in U_{\epsilon}$. 

If $(X_{t},\lambda_{t})\in U_{\epsilon/2}$, then it follows from \eqref{Eq:aoaosjsjsdjdjdhfgfhgfhdhdgfgf}:
\begin{align*}
&\tilde{F}((x_{*},\lambda_{*}),(Y_{t+1},\lambda_{t+1}))\\
&\leq \tilde{F}((x_{*},\lambda_{*}),(Y_{t+1},\lambda_{t+1}))+\gamma_{t}\psi_{t}+\gamma_{n}^{2}\tilde{C}.
\end{align*}
Thus for sufficiently large $t\in\nat$ s.t. $\gamma_{t}\psi_{t}+\gamma_{t}^{2}\tilde{C}< \epsilon/2$, we have $\tilde{F}(\VS(\X\times\real^{M}_{\geq 0}),(Y_{t+1},\lambda_{t+1}))<\epsilon$.
We are done by combining all the observations.
\end{proof}
The convergence $(X_{t})_{t}$ is now immediate:
\begin{corollary}
Suppose that the assumptions given in Theorem \ref{Lem:aoaosjsjshdhdjshssjjddd} holds and suppose that $\VS(\mathcal{Q},v),\VS(\X\times\real^{M}_{\geq 0},\tilde{v})\neq \emptyset$. Then $X_{n}\rightarrow \VS(\mathcal{Q},v)$.
\end{corollary}
\begin{proof}
Theorem \ref{Lem:aoaosjsjshdhdjshssjjddd} asserts that $(X_{t},\lambda_{t})\rightarrow\VS(\X\times \real_{\geq 0}^{M},\tilde{v})$. Proposition \ref{Prop:aiaisshshjdhdggddd} and the assumption $\VS(\X\times\real^{M}_{\geq 0},\tilde{v})\neq \emptyset$ implies that $\VS(\X\times \real_{\geq 0}^{M},\tilde{v})=\SOL(\X\times \real_{\geq 0}^{M},\tilde{v})$. Moreover Proposition \ref{Prop:aoaoajsjsksjjdddddd} asserts that $(\overline{x},\overline{\lambda})\in\SOL(\X\times \real_{\geq 0}^{M},\tilde{v})$ implies that $\overline{x}\in \SOL(\X\times \real_{\geq 0}^{M},v)$. Finally, since $\VS(\mathcal{Q},v)\neq\emptyset$ we have $\VS(\mathcal{Q},v)=\SOL(\mathcal{Q},v)$.
\end{proof}
At last, let us provide an example of sequences $(\gamma_{t})$ and $(\theta_{t})$ which fulfills the condition given in Theorem \ref{Lem:aoaosjsjshdhdjshssjjddd}:
\begin{remark}
Suppose that $\gamma_{t}=1/(t+1)$, then the condition \eqref{Eq:aiiahhssjjdhhdjjdjdsss} is fulfilled since $\sum_{k=1}^{\infty}\gamma_{k}=\infty$ and $\sum_{k=1}^{\infty}\gamma_{k}^2<\infty$. For $\theta_{t}=\delta\gamma_{t}$, we have:
\begin{equation*}
\gamma_{t}(2\theta_{t}^{2}+\tilde{C}_{2})-\tfrac{\theta_{t}}{2}=\gamma_{t}\left[2\delta\gamma_{t}^{2}+\tilde{C}_{2}-\tfrac{\delta}{2}\right].
\end{equation*} 
In case $\delta>2\tilde{C}_{2}$, we can find $c>0$ s.t.:
\begin{equation*}
\gamma_{t}(2\theta_{t}^{2}+\tilde{C}_{2})-\tfrac{\theta_{t}}{2}\leq\gamma_{t}\left[2\delta\gamma_{t}^{2}-c\right].
\end{equation*}
The latter is clearly negative for $t$ sufficiently large.
\end{remark}

\section{Discussions and Outlooks}
In this work we have introduce a novel semi-decentralized algorithm for concave games with coupled constraints based on mirror ascent and the method of augmented lagrangian. We provide a sufficient condition on the step size sequence and the degree of augmentation such that the algorithm converges to variationally stable Nash equilibrium. Specific choices of step-size sequence and augmentation sequence for that purpose is also provided. In particular, step size of order $\gamma_{t}=\mathcal{O}(1/t)$ leads to this desired state.

In the future work we plan also to investigate the case where the step-size - and augmentation sequence is adaptive. Moreover it is interesting to know whether the case where step size sequences of the agents differ. 

Another interesting line of work is to investigate whether the algorithm is robust toward random disturbance. That is to investigate how the algorithm performance if the feedback obtained by the agents is an unbiased martingale estimate of the gradient.

We also plan to derive based on the algorithm given in this work an algorithm which ensures not only compliance in the asymptotic region but also in the non-asymptotic region.


\bibliographystyle{IEEEtran}

\begin{thebibliography}{10}
	\providecommand{\url}[1]{#1}
	\csname url@samestyle\endcsname
	\providecommand{\newblock}{\relax}
	\providecommand{\bibinfo}[2]{#2}
	\providecommand{\BIBentrySTDinterwordspacing}{\spaceskip=0pt\relax}
	\providecommand{\BIBentryALTinterwordstretchfactor}{4}
	\providecommand{\BIBentryALTinterwordspacing}{\spaceskip=\fontdimen2\font plus
		\BIBentryALTinterwordstretchfactor\fontdimen3\font minus
		\fontdimen4\font\relax}
	\providecommand{\BIBforeignlanguage}[2]{{%
			\expandafter\ifx\csname l@#1\endcsname\relax
			\typeout{** WARNING: IEEEtran.bst: No hyphenation pattern has been}%
			\typeout{** loaded for the language `#1'. Using the pattern for}%
			\typeout{** the default language instead.}%
			\else
			\language=\csname l@#1\endcsname
			\fi
			#2}}
	\providecommand{\BIBdecl}{\relax}
	\BIBdecl
	
	\bibitem{Scutari2012}
	G.~Scutari, D.~P. Palomar, F.~Facchinei, and J.-S. Pang, \emph{Monotone Games
		for Cognitive Radio Systems}.\hskip 1em plus 0.5em minus 0.4em\relax London:
	Springer London, 2012, pp. 83--112.
	
	\bibitem{Mohsenian2010}
	A.~Mohsenian-Rad, V.~W.~S. Wong, J.~Jatskevich, R.~Schober, and A.~Leon-Garcia,
	``Autonomous demand-side management based on game-theoretic energy
	consumption scheduling for the future smart grid,'' \emph{IEEE Trans. on
		Smart Grid}, vol.~1, no.~3, pp. 320 -- 331, Dec. 2010.
	
	\bibitem{Saad2012}
	W.~Saad, Z.~Han, H.~V. Poor, and T.~Basar, ``Game-theoretic methods for the
	smart grid: An overview of microgrid systems, demand-side management, and
	smart grid communications,'' \emph{IEEE Sig. Proc. Mag.}, vol.~29, no.~5, pp.
	86 -- 105, 2012.
	
	\bibitem{Deng2014}
	R.~Deng, Z.~Yang, J.~Chen, N.~R. Asr, and M.~Chow, ``Residential energy
	consumption scheduling: A coupled-constraint game approach,'' \emph{IEEE
		Trans. on Smart Grid}, vol.~5, no.~3, pp. 1340 -- 1350, May 2014.
	
	\bibitem{Li2016}
	S.~Li, W.~Zhang, J.~Lian, and K.~Kalsi, ``Market-based coordination of
	thermostatically controlled loads part i: A mechanism design formulation,''
	\emph{IEEE Trans. on Pow. Sys.}, vol.~31, no.~2, pp. 1170 -- 1178, March
	2016.
	
	\bibitem{Li2016(2)}
	------, ``On reverse stackelberg game and optimal mean field control for a
	large population of thermostatically controlled loads,'' in \emph{2016 Am.
		Cont. Conf.}, July 2016, pp. 3545 -- 3550.
	
	\bibitem{Grammatico2015}
	S.~Grammatico, B.~Gentile, F.~Parise, and J.~Lygeros, ``A mean field control
	approach for demand side management of large populations of thermostatically
	controlled loads,'' in \emph{Proc. of the IEEE European Control Conference},
	2015.
	
	\bibitem{Ma2013}
	Z.~Ma, D.~S. Callaway, and I.~A. Hiskens, ``Decentralized charging control of
	large populations of plug-in electric vehicles,'' \emph{IEEE Trans. on Cont.
		Sys. Tech.}, vol.~21, no.~1, pp. 67 -- 78, Jan. 2013.
	
	\bibitem{Parise2014}
	F.~Parise, M.~Colombino, S.~Grammatico, and J.~Lygeros, ``Mean field
	constrained charging policy for large populations of plug-in electric
	vehicles,'' in \emph{53rd IEEE Conference on Decision and Control}, Dec.
	2014, pp. 5101 -- 5106.
	
	\bibitem{Ma2016}
	Z.~Ma, S.~Zou, L.~Ran, X.~Shi, and I.~A. Hiskens, ``Efficient decentralized
	coordination of large-scale plug-in electric vehicle charging,''
	\emph{Automatica}, vol.~69, pp. 35 -- 47, 2016.
	
	\bibitem{Grammatico20162}
	S.~Grammatico, ``Exponentially convergent decentralized charging control for
	large populations of plug-in electric vehicles,'' in \emph{2016 IEEE 55th
		Conference on Decision and Control (CDC)}, Dec. 2016, pp. 5775 -- 5780.
	
	\bibitem{Li2015}
	N.~Li, L.~Chen, and M.~A. Dahleh, ``Demand response using linear supply
	function bidding,'' \emph{IEEE Trans. on Smart Grid}, vol.~6, no.~4, pp.
	1827--1838, July 2015.
	
	\bibitem{Barrera2015}
	J.~Barrera and A.~Garcia, ``Dynamic incentives for congestion control,''
	\emph{IEEE Transactions on Automatic Control}, vol.~60, no.~2, pp. 299 --
	310, Feb. 2015.
	
	\bibitem{Belmaga2018}
	E.~V. Belmega, P.~Mertikopoulos, R.~Negrel, and L.~Sanguinetti, ``Online convex
	optimization and no-regret learning: Algorithms, guarantees and
	applications,'' \emph{arXiv:1804.04529}, 2018.
	
	\bibitem{Shalev-Shwartz2012}
	S.~Shalev-Shwartz, ``Online learning and online convex optimization,''
	\emph{Foundations and Trends in Machine Learning}, vol.~4, 2012.
	
	\bibitem{Low1}
	S.~H. Low and D.~E. Lapsley, ``Optimization flow control, i: Basic algorithm
	and convergence,'' \emph{IEEE/ACM Transactions on Networking}, 1999.
	
	\bibitem{Mertikopoulos2018}
	P.~Mertikopoulos and Z.~Zhou, ``Learning in games with continuous action sets
	and unknown payoff functions,'' \emph{Mathematical Programming}, Mar. 2018.
	
	\bibitem{Paccagnan2017}
	D.~Paccagnan, B.~G.~G. Parise, M.~Kamgarpour, and J.Lygeros, ``Nash and wardrop
	equilibria in aggregative games with coupling constraints,'' 2017,
	arXiv:1702.08789.
	
	\bibitem{Facchinei1}
	F.~Facchinei and J.-S. Pang, \emph{Finite-Dimensional Variational Inequalities
		and Complementarity Problems}.\hskip 1em plus 0.5em minus 0.4em\relax
	Springer-Verlag New York, 2003.
	
	\bibitem{Mahdavi1}
	M.~Mahdavi, R.~Jin, and T.~Yang, ``Trading regret for efficiency: Online convex
	optimization with long term constraints,'' \emph{J. Mach. Learn. Res.},
	vol.~13, no.~1, pp. 2503 -- 2528, Jan. 2012.
	
	\bibitem{Rockafellar1970}
	R.~T. Rockafellar, \emph{Convex Analysis}.\hskip 1em plus 0.5em minus
	0.4em\relax Princeton University Press, 1970.
	
	\bibitem{Rockafellar1998}
	R.~T. Rockafellar and R.~J.~B. Wets, \emph{Variational Analysis}, ser. A Ser.
	of Comp. Stud. in Math.\hskip 1em plus 0.5em minus 0.4em\relax
	Springer-Verlag, 1998, vol. 317.
	
	\bibitem{Mertikopoulos2016}
	P.~Mertikopoulos and W.~H. Sandholm., ``Learning in games via reinforcement and
	regularization,'' \emph{Math. of Op. Res.}, vol.~14, no.~1, pp. 124 -- 143,
	2016.
	
\end{thebibliography}

\end{document}